\documentclass[a4paper]{article}

\usepackage[dvipdfmx]{hyperref}
\usepackage[dvipdfmx]{graphicx,xcolor}
\usepackage{amsmath,amsfonts,amssymb,amsthm}

\title{Interleaved sequences of geometric sequences binarized with Legendre symbol of two types}
\author{Kazuyoshi Tsuchiya \thanks{The author is with the Koden Electronics Co., Ltd., Ota-ku, 146-0095 Japan} \and
Yasuyuki Nogami \thanks{The author is with the Graduate School of Natural Science and Technology, Okayama University,
Okayama-shi, 700-8530 Japan}  \and
Satoshi Uehara \thanks{The author is with the Faculty of Environmental Engineering, the University of Kitakyushu,
Kitakyushu-shi, 808-0135 Japan}
\date{}}

\newtheorem{theorem}{Theorem}
\newtheorem{proposition}[theorem]{Proposition}
\newtheorem{lemma}[theorem]{Lemma}
\newtheorem{corollary}[theorem]{Corollary}

\newtheorem{example}[theorem]{Example}
\newtheorem{remark}[theorem]{Remark}

\begin{document}
\maketitle
\begin{abstract}
A pseudorandom number generator is widely used in cryptography.
A cryptographic pseudorandom number generator is required to generate pseudorandom numbers
which have good statistical properties as well as unpredictability.
An m-sequence is a linear feedback shift register sequence with maximal period over a finite field.
M-sequences have good statistical properties, however we must nonlinearize m-sequences for cryptographic purposes.
A geometric sequence is a binary sequence given by applying a nonlinear feedforward function to an m-sequence.
Nogami, Tada and Uehara proposed a geometric sequence
whose nonlinear feedforward function is given by the Legendre symbol.
They showed the geometric sequences have good properties for
the period, periodic autocorrelation and linear complexity.
However, the geometric sequences do not have the balance property.
In this paper, we introduce geometric sequences of two types and
show some properties of interleaved sequences of the geometric sequences of two types.
These interleaved sequences have the balance property
and double the period of the geometric sequences by the interleaved structure.
Moreover, we show correlation properties and linear complexity of the interleaved sequences.
A key of our observation is that the second type geometric sequence is the complement of
the left shift of the first type geometric sequence by half-period positions.
\end{abstract}
%\begin{keywords}
%pseudorandom number generator, geometric sequence, balance property, interleaved sequence.
%\end{keywords}

\section{Introduction}

In cryptography, a pseudorandom number generator is used to generate a private key, a public key,
a session key, a keystream and so on.
A cryptographic pseudorandom number generator is required to generate pseudorandom numbers
which have good statistical properties as well as unpredictability.

A linear feedback shift register (LFSR) \cite{Golomb,Goresky-Klapper} is a fast pseudorandom number generator
which generates well-distributed sequences with long period.
In particular, an m-sequence is an LFSR sequence with maximal period over a finite field.
However, we must nonlinearize m-sequences for cryptographic purposes.
A geometric sequence \cite{Chan-Games} is a binary sequence
given by applying a nonlinear feedforward function to an m-sequence.
As geometric sequences,
GMW sequences \cite{Gordon-Mills-Welch,Scholtz-Welch},
cascaded GMW sequences \cite{Klapper-Chan-Goresky,Gong1996}
and so on have been well known.

Nogami, Tada and Uehara \cite{Nogami-Tada-Uehara:IEICE} proposed a geometric sequence
whose nonlinear feedforward function is given by the Legendre symbol.
They showed the geometric sequences have good properties for period, periodic autocorrelation and linear complexity.
Moreover, the studies on generalization and extension of the sequences
\cite{Ino-Nogami-Uehara,Koike-Nogami-Tsuchiya-Uehara,Arshad-Nogami-Ogawa-Ino-Uehara-Morelos-Zaragoza-Tsuchiya,Ogawa-Arshad-Nogami-Uehara-Tsuchiya-Morelos-Zaragoza,Nogami-Uehara-Tsuchiya-Begum-Ino-Morelos-Zaragoza}
have made progress in recent years.
On the other hand, the sequences do not have the balance property.
It may become an obstruction for the fast implementation
proposed by Nogami, Tada and Uehara \cite{Nogami-Tada-Uehara:ISITA}.

The purpose of this paper is to construct sequences from the geometric sequences, which have the balance property.
Since the geometric sequences are described by the Legendre symbol,
the geometric sequences have a property similar to that of Legendre sequences \cite{Zierler}.
Legendre sequences of two types were defined and
properties of the interleaved sequences of Legendre sequences of two types were studied
by Tang and Gong \cite{Tang-Gong} and some other researchers.
In this paper, we introduce geometric sequences of two types and
show properties of interleaved sequences of the geometric sequences of two types.
More concretely, we show the period, correlation properties and linear complexity, and
show that the sequences have the balance property.
A key of our observation is that the second type geometric sequence is the complement of
the left shift of the first type geometric sequence by half-period positions.

This paper is organized as follows:
In Sect.\,\ref{section:Preliminaries},
we recall the definition of periodic autocorrelation, periodic cross-correlation and linear complexity.
In Sect.\,\ref{section:Geometric sequences binarized with Legendre symbol},
we introduce the geometric sequences of two types and show some properties of the sequences.
In Sect.\,\ref{section:Interleaved sequences},
we propose interleaved sequences of the geometric sequences of two types and
show the period, correlation properties and linear complexity of the interleaved sequences.
Finally, we describe some conclusions in Sect.\ref{section:Conclusion}.

We give some notations.
For a prime power $q$, $\mathbb{F}_q$ denotes the finite field with $q$ elements.
For a polynomial $f ( x )$ over a field $\mathbb{F}$, $\mathrm{deg} \, f$ denotes the degree of $f ( x )$.
For a prime number $p$ and an integer $a$, $(a / p)$ denotes the Legendre symbol.
For a finite field $\mathbb{F}_p$ and the extension field $\mathbb{F}_{p^m}$ of degree $m$ of $\mathbb{F}_p$,
$\mathrm{Tr}_{\mathbb{F}_{p^m} / \mathbb{F}_p } : \mathbb{F}_{p^m} \rightarrow \mathbb{F}_p$ denotes the trace map,
namely,
$\mathrm{Tr}_{\mathbb{F}_{p^m} / \mathbb{F}_p} ( \alpha ) = \alpha + \alpha^p + \cdots + \alpha^{p^{m - 1}}
\in \mathbb{F}_p$ for $\alpha \in \mathbb{F}_{p^m}$.
For a finite set $A$, $\# A$ denotes the number of elements in $A$.
For an integer $n > 1$ and an integer $a$, $a ~\underline{\mathrm{mod}}~ n$ denotes
the integer $u \in \{ 1, 2, \dots ,n \}$ such that $u \equiv a \mathrm{~mod~} n$. 

\section{Preliminaries}
\label{section:Preliminaries}

In this section, we recall the definition of periodic autocorrelation, periodic cross-correlation and linear complexity.
For details, see \cite[10.3, 10.4]{Mullen-Panario}.

\subsection{Periodic autocorrelation and periodic cross-correlation}

Let $S^{(1)} = ( s^{(1)}_n )_{n \geq 0}$ and $S^{(2)} = ( s^{(2)}_n )_{n \geq 0}$
be $N$-periodic sequences over $\mathbb{F}_2$.
For $0 \leq \tau \leq N - 1, \tau \in \mathbb{Z}$,
$R_{S^{(1)}, S^{(2)}} ( \tau )$ is defined as
\begin{equation*}
R_{S^{(1)}, S^{(2)}} ( \tau ) =
\sum_{i = 0}^{N - 1} (-1)^{s^{(1)}_i + s^{(2)}_{i + \tau}}.
\end{equation*}
Note that $s^{(1)}_i + s^{(2)}_{i + \tau}$ means an addition in $\mathbb{F}_2$.
If $S^{(1)} \neq S^{(2)}$,
then $R_{S^{(1)}, S^{(2)}} ( \tau )$ is called
the periodic cross-correlation of $S^{(1)}$ and $S^{(2)}$.
If $S := S^{(1)} = S^{(2)}$,
then $R_{S} ( \tau ) := R_{S, S} ( \tau )$ is called
the periodic autocorrelation of $S$.
By the definition,  $R_{S} ( 0 ) = N$.

\subsection{Linear complexity}

Let $q$ be a prime power.
Let $S = ( s_n )_{n \geq 0}$ be a sequence over $\mathbb{F}_q$.
The linear complexity $L(S)$ of $S$ is the length $L$ of the shortest linear recurrence relation
\begin{equation*}
s_{n + L} = a_{L - 1} s_{n + L - 1} + \cdots + a_{0} s_{n}, \quad n \geq 0
\end{equation*}
for some $a_0, \dots ,a_{L - 1} \in \mathbb{F}_q$,
and the minimal polynomial $m(x)$ of $S$ is the polynomial
$x^L - a_{L - 1} x^{L - 1} - \cdots - a_0  \in \mathbb{F}_q [x]$. 

Now, assume that $S$ is an $N$-periodic sequence.
$S(x)$ denotes the polynomial $s_0 + s_1 x + \cdots + s_{N - 1} x^{N - 1} \in \mathbb{F}_q [x]$.
Then $m(x) = (x^N - 1) / \mathrm{GCD} \, (x^N - 1, S (x))$ and
$L(S) = N - \mathrm{deg} \, \mathrm{GCD} \ (x^N - 1, S (x))$.

\section{Geometric sequences binarized with Legendre symbol}
\label{section:Geometric sequences binarized with Legendre symbol}

Nogami, Tada and Uehara \cite{Nogami-Tada-Uehara:IEICE} proposed a geometric sequence
whose nonlinear feedforward function is given by the Legendre symbol.
In this section, we survey properties of the geometric sequences
and define the geometric sequences of two types as the Legendre sequences (See \cite[III, B]{Tang-Gong}).
Let $p > 2$ be a prime number and $m > 1$ an integer. 
Let $\omega \in \mathbb{F}_{p^m}$ be a primitive element in $\mathbb{F}_{p^m}$.

\subsection{The geometric sequences of the first type}

The geometric sequence $T^{(1)} = ( t^{(1)}_n )_{n \geq 0}$ of the first type is defined as
\begin{equation}
t^{(1)}_n =
\left\{
\begin{array}{ll}
0 & \mbox{if~} \left( \frac{\mathrm{Tr}_{\mathbb{F}_{p^m} / \mathbb{F}_p}( \omega^n )}{p} \right) = 1 \\
1 & \mbox{if~} \left( \frac{\mathrm{Tr}_{\mathbb{F}_{p^m} / \mathbb{F}_p}( \omega^n )}{p} \right) = - 1 \\
0 & \mbox{if~} \left( \frac{\mathrm{Tr}_{\mathbb{F}_{p^m} / \mathbb{F}_p}( \omega^n )}{p} \right) = 0
\end{array}
\right.
, \quad n \geq 0.
\label{equation:1st type NTU}
\end{equation}

$T^{(1)}$ is a periodic sequence of period $N := 2 (p^m - 1) / (p - 1)$.
Note that $N$ is an even integer.
For $0 \leq \tau \leq N - 1, \tau \in \mathbb{Z}$, the autocorrelation of $T^{(1)}$ is
\begin{equation}
R_{T^{(1)}} (\tau) =
\left\{
\begin{array}{ll}
N = \frac{2 (p^m - 1)}{p - 1} & \mbox{if~} \tau = 0 \\
N_1 := - 2 p^{m - 1} + \frac{2 (p^{m - 1} - 1)}{p - 1} & \mbox{if~} \tau = \frac{N}{2} \\
N_2 := \frac{2 (p^{m - 2} - 1)}{p - 1} & \mbox{otherwise}.
\end{array}
\right.
\label{equation:AC of NTU}
\end{equation}
The linear complexity of $T^{(1)}$ is $L( T^{(1)} ) = N = 2 (p^m - 1) / (p - 1)$ (See \cite{Chan-Games} for arbitrary $m$).
Hence $L( T^{(1)} )$ attains a maximal value.
On the other hand, $T^{(1)}$ does not have the balance property.
In fact,
the number of zeros in one period of $T^{(1)}$ is $p^{m - 1} + 2 (p^{m - 1} - 1) / (p - 1)$, and
the number of ones in one period of $T^{(1)}$ is $p^{m - 1}$.

\subsection{The geometric sequences of the second type}

We define a geometric sequence $T^{(2)} = ( t^{(2)}_n )_{n \geq 0}$ of the second type as
\begin{equation}
t^{(2)}_n =
\left\{
\begin{array}{ll}
0 & \mbox{if~} \left( \frac{\mathrm{Tr}_{\mathbb{F}_{p^m} / \mathbb{F}_p}( \omega^n )}{p} \right) = 1 \\
1 & \mbox{if~} \left( \frac{\mathrm{Tr}_{\mathbb{F}_{p^m} / \mathbb{F}_p}( \omega^n )}{p} \right) = - 1 \\
1 & \mbox{if~} \left( \frac{\mathrm{Tr}_{\mathbb{F}_{p^m} / \mathbb{F}_p}( \omega^n )}{p} \right) = 0
\end{array}
\right.
, \quad n \geq 0.
\label{equation:2nd type NTU}
\end{equation}

\begin{lemma}
For any $n \geq 0 \in \mathbb{Z}$, $t^{(2)}_n = t^{(1)}_{n + N / 2} + 1 \in \mathbb{F}_2.$
\label{lemma:relation between type 1 and type 2}
\end{lemma}

\begin{proof}
Since $\omega^{N / 2} \in \mathbb{F}_p$ is a primitive element in $\mathbb{F}_p$,
\begin{equation*}
\left( \frac{\mathrm{Tr}_{\mathbb{F}_{p^m} / \mathbb{F}_p}( \omega^{n + N / 2} )}{p} \right) =
\left( \frac{\omega^{N / 2} \cdot \mathrm{Tr}_{\mathbb{F}_{p^m} / \mathbb{F}_p}( \omega^{n} )}{p} \right)
= - \left( \frac{\mathrm{Tr}_{\mathbb{F}_{p^m} / \mathbb{F}_p}( \omega^{n} )}{p} \right).
\end{equation*}
Hence $t^{(1)}_{n} = t^{(1)}_{n + N / 2} + 1$
for any $n \geq 0$ such that $\mathrm{Tr}_{\mathbb{F}_{p^m} / \mathbb{F}_p}( \omega^{n} ) \neq 0$.
Noting that
$\mathrm{Tr}_{\mathbb{F}_{p^m} / \mathbb{F}_p}( \omega^{n + N / 2} ) = 0$ if and only if
$\mathrm{Tr}_{\mathbb{F}_{p^m} / \mathbb{F}_p}( \omega^{n} ) = 0$,
it follows that
\begin{eqnarray*}
t^{(2)}_{n} & = &
\left\{
\begin{array}{ll}
t^{(1)}_{n} & \mbox{if~} \mathrm{Tr}_{\mathbb{F}_{p^m} / \mathbb{F}_p}( \omega^{n} ) \neq 0 \\
t^{(1)}_{n} + 1 & \mbox{if~} \mathrm{Tr}_{\mathbb{F}_{p^m} / \mathbb{F}_p}( \omega^{n} ) = 0 
\end{array}
\right. \\
& = & t^{(1)}_{n + N / 2} + 1.
\end{eqnarray*}
\end{proof}

\begin{corollary}
The period, periodic autocorrelation and linear complexity of $T^{(2)}$
are the same as those of $T^{(1)}$, respectively.
\label{corollary:relation between type 1 and type 2}
\end{corollary}

\begin{proposition}
For $0 \leq \tau \leq N - 1, \tau \in \mathbb{Z}$,
$R_{T^{(1)}, T^{(2)}} (\tau) = - R_{T^{(1)}} \left( \tau + N / 2 \right)$.
\label{proposition:periodic cross-correlation of type 1 and type 2}
\end{proposition}

\begin{proof}
By Lemma \ref{lemma:relation between type 1 and type 2},
\begin{eqnarray*}
R_{T^{(1)}, T^{(2)}} (\tau)
& = & \sum_{i = 0}^{N - 1} (-1)^{t^{(1)}_i + t^{(2)}_{i + \tau}}
= \sum_{i = 0}^{N - 1} (-1)^{t^{(1)}_i + t^{(1)}_{i + \tau + N / 2} + 1} \\
& = & \# \left\{ i \mid t^{(1)}_i = t^{(1)}_{i + \tau + N / 2} + 1, \, 0 \leq i \leq N - 1 \right\} \\
& & - \# \left\{ i \mid t^{(1)}_i \neq t^{(1)}_{i + \tau + N / 2} + 1, \, 0 \leq i \leq N - 1 \right\}  \\
& = & \# \left\{ i \mid t^{(1)}_i \neq t^{(1)}_{i + \tau + N / 2}, \, 0 \leq i \leq N - 1 \right\} \\
& & - \# \left\{ i \mid t^{(1)}_i = t^{(1)}_{i + \tau + N / 2}, \, 0 \leq i \leq N - 1 \right\} \\
& = & - R_{T^{(1)}} \left( \tau + N / 2 \right).
\end{eqnarray*}
\end{proof}

\section{Interleaved sequences}
\label{section:Interleaved sequences}

In order to construct a sequence that has the balance property, we introduce an interleaved sequence.
Gong \cite{Gong1995} introduced the concept of interleaved sequences that were obtained by merging sequences, and
sequence interleaving was widely used to improve the balance property of existing sequences
such as in \cite{Tang-Ding,Gu-Hwang-Han-Kim-Jin,Gong2002,Edemskiy,Chung-Yang,Zeng-Zeng-Zhang-Xuan,Yu-Gong}.

In this section, we propose interleaved sequences of geometric sequences of two types,
which were introduced in the previous section.
From Lemma \ref{lemma:relation between type 1 and type 2},
the interleaved sequence has the balance property by the interleaved structure.
Furthermore, we show correlation properties and linear complexity of the interleaved sequences.

\subsection{Interleaved sequences and left cyclic shift sequences}

For a family $\mathbb{S} = \{ S^{(i)} = ( s^{(i)}_{n} )_{ n \geq 0 } \mid 0 \leq i \leq T - 1 \}$
of $N$-periodic sequences, the sequence $U = ( u_n )_{ n \geq 0 }$ defined as
\begin{equation*}
u_n = s^{(i)}_{j} \mbox{~if~} n = j \times T + i, \, 0 \leq i \leq T - 1, \, 0 \leq j
\end{equation*}
is called an interleaved sequence of $\mathbb{S}$.

For an $N$-periodic sequence $S = ( s_{n} )_{ n \geq 0 }$ and $e \in \{ 0, 1, \dots , N - 1 \}$,
the sequence $L^{e} (S) = ( L^e (s)_{n} )_{ n \geq 0 }$ defined as
\begin{equation*}
L^e (s)_{n} = s_{n + e}, \quad n \geq 0
\end{equation*}
is called a left cyclic shift sequence.

\subsection{Interleaved sequences of the geometric sequences of two types}

Let $p > 2$ be a prime number and $m > 1$ an integer. 
Let $\omega \in \mathbb{F}_{p^m}$ be a primitive element in $\mathbb{F}_{p^m}$.
Let $T^{(1)}$ and $T^{(2)}$ be the sequences defined as (\ref{equation:1st type NTU}) and (\ref{equation:2nd type NTU})
with respect to $p, m$ and $\omega$, respectively.
Then $N := 2 (p^m - 1) / (p - 1)$ is the period of $T^{(1)}$ and $T^{(2)}$.

For $e \in \{ 0, 1, \dots , N - 1 \}$,
we define a sequence $S^{e} = ( s^{e}_{n} )_{ n \geq 0 }$
as the interleaved sequence of $\{ T^{(1)}, L^{e} (T^{(2)}) \}$.

\begin{example}
Let $p = 5$ and $m = 2$, so that $N = 12$.
Let $\alpha$ be a root of the irreducible polynomial $x^2 + 2 x + 3$, and
$\omega = 4 \alpha \in \mathbb{F}_p ( \alpha )$.
Then $T^{(1)}$ and $T^{(2)}$ are given as
\begin{eqnarray*}
T^{(1)} & = & ( 1, 1, 1, 0, 0, 1, 0, 0, 0, 0, 1, 0, \dots ), \\
T^{(2)} & = & ( 1, 1, 1, 1, 0, 1, 0, 0, 0, 1, 1, 0, \dots ),
\end{eqnarray*}
respectively.
Put $e = 4$.
Then $S^{e}$ is given as
\begin{equation*}
S^{4} = ( 1, 0, 1, 1, 1, 0, 0, 0, 0, 0, 1, 1, 0, 1, 0, 0, 0, 1, 0, 1, 1, 1, 0, 1, \dots ).
\end{equation*}
\end{example}

\begin{example}
Let $p = 3$ and $m = 3$, so that $N = 26$.
Let $\alpha$ be a root of the irreducible polynomial $x^3 + 2 x^2 + 1$, and
$\omega = 2 \alpha^2 \in \mathbb{F}_p ( \alpha )$.
Then $T^{(1)}$ and $T^{(2)}$ are given as
\begin{eqnarray*}
T^{(1)} & = & ( 0, 1, 0, 1, 1, 0, 0, 0, 1, 1, 1, 0, 1, 0, 0, 0, 0, 0, 0, 0, 1, 0, 0, 0, 1, 0, \dots ), \\
T^{(2)} & = & ( 1, 1, 1, 1, 1, 1, 1, 0, 1, 1, 1, 0, 1, 1, 0, 1, 0, 0, 1, 1, 1, 0, 0, 0, 1, 0, \dots ),
\end{eqnarray*}
respectively.
Put $e = 17$.
Then $S^{e}$ is given as
\begin{eqnarray*}
S^{17} & = & ( 0, 0, 1, 1, 0, 1, 1, 1, 1, 0, 0, 0, 0, 0, 0, 1, 1, 0, 1, 1, 1, 1, 0, 1, 1, 1, \\
& & 0, 1, 0, 1, 0, 1, 0, 0, 0, 1, 0, 1, 0, 1, 1, 0, 0, 1, 0, 1, 0, 0, 1, 1, 0, 0, \dots ).
\end{eqnarray*}
\end{example}

\subsection{Periodic autocorrelation and periodic cross-correlation of the proposed sequences}

We show a lemma to give the periodic autocorrelation and periodic cross-correlation of the proposed sequences.

\begin{lemma}
If $\tau = 2 \tau_0, \, 0 \leq \tau_0 \leq N - 1, \tau_0 \in \mathbb{Z}$, then
\begin{equation*}
R_{ S^{e_1}, S^{e_2} } ( \tau ) = R_{ T^{(1)} } ( \tau_0 ) + R_{ T^{(1)} } ( e_2 - e_1 + \tau_0 ).
\end{equation*}
If $\tau = 2 \tau_0 + 1, \, 0 \leq \tau_0 \leq N - 1, \tau_0 \in \mathbb{Z}$, then
\begin{equation*}
R_{ S^{e_1}, S^{e_2} } ( \tau ) = - R_{ T^{(1)} } \left( e_2 + \tau_0 + \frac{N}{2} \right)
- R_{ T^{(1)} } \left( e_1 - \tau_0 - 1 + \frac{N}{2} \right).
\end{equation*}
\label{lemma:description of R(tau) by T^(1)}
\end{lemma}

\begin{proof}
Assume that $\tau = 2 \tau_0, \, 0 \leq \tau_0 \leq N - 1, \tau_0 \in \mathbb{Z}$.
By Corollary \ref{corollary:relation between type 1 and type 2},
\begin{eqnarray*}
R_{ S^{e_1}, S^{e_2} } ( \tau ) & = & R_{ T^{(1)} } ( \tau_0 ) + R_{ T^{(2)} } ( e_2 - e_1 + \tau_0 ) \\
& = & R_{ T^{(1)} } ( \tau_0 ) + R_{ T^{(1)} } ( e_2 - e_1 + \tau_0 ).
\end{eqnarray*}

Assume that $\tau = 2 \tau_0 + 1, \, 0 \leq \tau_0 \leq N - 1, \tau_0 \in \mathbb{Z}$.
By Proposition \ref{proposition:periodic cross-correlation of type 1 and type 2},
\begin{eqnarray*}
R_{ S^{e_1}, S^{e_2} } ( \tau )
& = & R_{ T^{(1)}, T^{(2)} } ( e_2 + \tau_0 ) + R_{ T^{(2)}, T^{(1)} } ( \tau_0 + 1 - e_1 ) \\
& = & R_{ T^{(1)}, T^{(2)} } ( e_2 + \tau_0 ) + R_{ T^{(1)}, T^{(2)} } ( e_1 - \tau_0 - 1 ) \\
& = & - R_{ T^{(1)} } \left( e_2 + \tau_0 + \frac{N}{2} \right)
- R_{ T^{(1)} } \left( e_1 - \tau_0 - 1 + \frac{N}{2} \right).
\end{eqnarray*}
\end{proof}

The periodic autocorrelation of the proposed sequence is given as follows:

\begin{theorem}
Let $p > 2$ be a prime number and $m > 1$ an integer. 
Let $\omega \in \mathbb{F}_{p^m}$ be a primitive element in $\mathbb{F}_{p^m}$.
Let $T^{(1)}$ and $T^{(2)}$ be the sequences defined as (\ref{equation:1st type NTU}) and (\ref{equation:2nd type NTU})
with respect to $p, m$ and $\omega$, respectively.
For $e \in \{ 0, 1, \dots , N - 1 \}$,
let $S^{e}$ be the interleaved sequence of $\{ T^{(1)}, L^{e} (T^{(2)}) \}$.
Put $N = 2 (p^m - 1) / (p - 1)$.
Then, for $\tau = 2 \tau_0, \, 0 \leq \tau_0 \leq N - 1, \tau_0 \in \mathbb{Z}$,
the autocorrelation $R_{ S^{e} } ( \tau )$ of $S^{e}$ is given as
\begin{equation}
\left\{
\begin{array}{ll}
2 N & \mbox{if~} \tau_0 = 0 \\
2 N_1 & \mbox{if~} \tau_0 = \frac{N}{2} \\
2 N_2 & \mbox{otherwise}.
\end{array}
\right.
\label{equation:periodic autocorrelation for even tau}
\end{equation}
For $\tau = 2 \tau_0 + 1, \, 0 \leq \tau_0 \leq N - 1, \tau_0 \in \mathbb{Z}$,
the autocorrelation $R_{ S^{e} } ( \tau )$ of $S^{e}$ is given as
\begin{equation}
\left\{
\begin{array}{ll}
- N - N_2 & \mbox{if~} \tau_0 \equiv -e - \frac{N}{2}, e - 1 + \frac{N}{2} \mathrm{~mod~} N \\
- N_1 - N_2 & \mbox{if~} \tau_0 \equiv - e, e - 1 \mathrm{~mod~} N \\
- 2 N_2 & \mbox{otherwise}
\end{array}
\right.
\label{equation:periodic autocorrelation for odd tau 1}
\end{equation}
if $2 e \not\equiv 1 - N / 2 \mathrm{~mod~} N$, and
\begin{equation}
\left\{
\begin{array}{ll}
- N - N_1 & \mbox{if~} \tau_0 \equiv - e, e - 1 \mathrm{~mod~} N \\
- 2 N_2 & \mbox{otherwise}
\end{array}
\right.
\label{equation:periodic autocorrelation for odd tau 2}
\end{equation}
if $2 e \equiv 1 - N / 2 \mathrm{~mod~} N$.
Here $N_1$ and $N_2$ are defined as in (\ref{equation:AC of NTU}).
\label{theorem:periodic autocorrelation}
\end{theorem}

\begin{proof}
Assume that $\tau = 2 \tau_0, \, 0 \leq \tau_0 \leq N - 1, \tau_0 \in \mathbb{Z}$.
Since $R_{ S^{e} } ( \tau ) = 2 R_{ T^{(1)} } ( \tau_0 )$
by Lemma \ref{lemma:description of R(tau) by T^(1)},
(\ref{equation:periodic autocorrelation for even tau}) is satisfied.

Assume that $\tau = 2 \tau_0 + 1, \, 0 \leq \tau_0 \leq N - 1, \tau_0 \in \mathbb{Z}$.
Then we have
$R_{ S^{e} } ( \tau )
= - R_{ T^{(1)} } \left( e + \tau_0 + N / 2 \right) - R_{ T^{(1)} } \left( e - \tau_0 - 1 + N / 2 \right)$
by Lemma \ref{lemma:description of R(tau) by T^(1)}.
On the other hand,
\begin{equation*}
R_{ T^{(1)} } \left( e + \tau_0 + \frac{N}{2} \right) =
\left\{
\begin{array}{ll}
N & \mbox{if~} \tau_0 \equiv - e - \frac{N}{2} \mathrm{~mod~} N \\
N_1 & \mbox{if~} \tau_0 \equiv - e \mathrm{~mod~} N \\
N_2 & \mbox{otherwise}
\end{array}
\right.
\end{equation*}
and
\begin{equation*}
R_{ T^{(1)} } \left( e - \tau_0 - 1 + \frac{N}{2} \right) =
\left\{
\begin{array}{ll}
N & \mbox{if~} \tau_0 \equiv e - 1 + \frac{N}{2} \mathrm{~mod~} N \\
N_1 & \mbox{if~} \tau_0 \equiv e - 1 \mathrm{~mod~} N \\
N_2 & \mbox{otherwise}.
\end{array}
\right.
\end{equation*}
If $2 e \not\equiv 1 - N / 2 \mathrm{~mod~} N$,
then $- e - N / 2, - e, e - 1 + N / 2, e - 1$ are incongruent modulo $N$.
Hence we have (\ref{equation:periodic autocorrelation for odd tau 1}).
If $2 e \equiv 1 - N / 2 \mathrm{~mod~} N$,
then $- e - N / 2 \equiv e - 1 \mathrm{~mod~} N$ and $- e \equiv e - 1 + N / 2 \mathrm{~mod~} N$.
Hence we have (\ref{equation:periodic autocorrelation for odd tau 2}).
\end{proof}

\begin{corollary}
$S^{e}$ has period $2 N = 4 (p^m - 1) / (p - 1)$.
\end{corollary}

By the interleaved structure, $S^{e}$ has the balance property,
namely the number of zeros and ones in one period of $S^{e}$ are $N$.

\begin{remark}
If $m$ is even, then $2 e \not\equiv 1 - N / 2 \mathrm{~mod~} N$.
\end{remark}

\begin{example}
Let $p = 11$ and $m = 2$, so that $N = 24$.
Let $\alpha$ be a root of the irreducible polynomial $x^2 + 7 x + 2$, and
$\omega = 9 + 2 \alpha \in \mathbb{F}_p ( \alpha )$.
Figure \ref{figure:AC of IS of NTU sequences in the case of p = 11, m = 2, e = 17}
describes the graph of the periodic autocorrelation of $S^{17}$.
\begin{figure}[htbp]
\begin{center}
\includegraphics[scale=0.85]{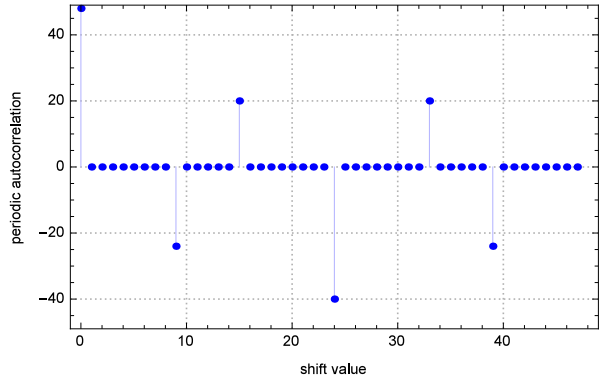}
\end{center}
\caption{The graph of the periodic autocorrelation of $S^{17}$ in the case of $p = 11, m = 2$.}
\label{figure:AC of IS of NTU sequences in the case of p = 11, m = 2, e = 17}
\end{figure}
\end{example}

\begin{example}
Let $p = 5$ and $m = 3$, so that $N = 62$.
Let $\alpha$ be a root of the irreducible polynomial $x^3 + 3 x^2 + 2 x + 3$, and
$\omega = 1 + \alpha + 2 \alpha^2 \in \mathbb{F}_p ( \alpha )$.
Figure \ref{figure:AC of IS of NTU sequences in the case of p = 5, m = 3, e = 25} and
Fig.\,\ref{figure:AC of IS of NTU sequences in the case of p = 5, m = 3, e = 16}
describe the graph of the periodic autocorrelation of $S^{25}$ and $S^{16}$, respectively.
\begin{figure}[htbp]
\begin{center}
\includegraphics[scale=0.85]{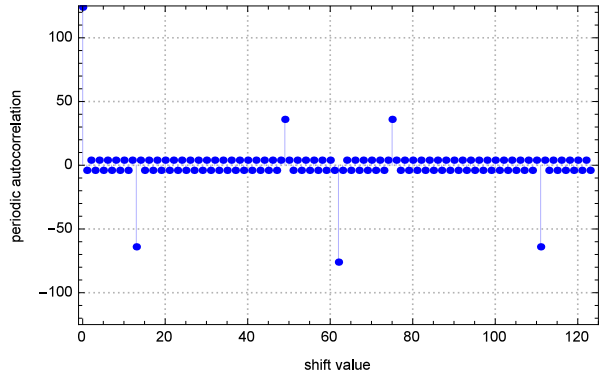}
\end{center}
\caption{The graph of the periodic autocorrelation of $S^{25}$ in the case of $p = 5, m = 3$.}
\label{figure:AC of IS of NTU sequences in the case of p = 5, m = 3, e = 25}
\end{figure}
\begin{figure}[htbp]
\begin{center}
\includegraphics[scale=0.85]{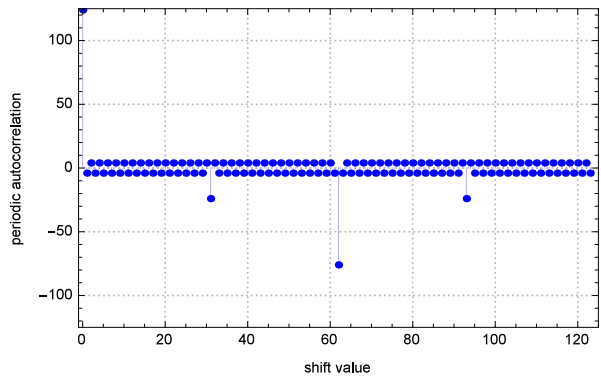}
\end{center}
\caption{The graph of the periodic autocorrelation of $S^{16}$ in the case of $p = 5, m = 3$.}
\label{figure:AC of IS of NTU sequences in the case of p = 5, m = 3, e = 16}
\end{figure}
\end{example}

Now we show the cross-correlation distribution of two proposed sequences that have different shift widths.

\begin{theorem}
Let $p > 2$ be a prime number and $m > 1$ an integer. 
Let $\omega \in \mathbb{F}_{p^m}$ be a primitive element in $\mathbb{F}_{p^m}$.
Let $T^{(1)}$ and $T^{(2)}$ be the sequences defined as (\ref{equation:1st type NTU}) and (\ref{equation:2nd type NTU}) with respect to $p, m$ and $\omega$, respectively.
For $e_1, e_2 \in \{ 0, 1, \dots , N - 1 \}, \, e_1 < e_2$,
let $S^{e_1}$ and $S^{e_2}$ be the interleaved sequences of
$\{ T^{(1)}, L^{e_1} (T^{(2)}) \}$ and $\{ T^{(1)}, L^{e_2} (T^{(2)}) \}$, respectively.
Put $N = 2 (p^m - 1) / (p - 1)$.
Then, for $\tau = 2 \tau_0, \, 0 \leq \tau_0 \leq N - 1, \tau_0 \in \mathbb{Z}$,
the cross-correlation $R_{S^{e_1}, S^{e_2}} (\tau)$ of $S^{e_1}$ and $S^{e_2}$ is given as
\begin{equation*}
\left\{
\begin{array}{ll}
N + N_2 & \mbox{if~} \tau_0 \equiv 0, e_1 - e_2 \mathrm{~mod~} N \\
N_1 + N_2 & \mbox{if~} \tau_0 \equiv \frac{N}{2}, e_1 - e_2 + \frac{N}{2} \mathrm{~mod~} N \\
2 N_2 & \mbox{otherwise}
\end{array}
\right.
\end{equation*}
if $e_2 - e_1 \not\equiv N / 2 \mathrm{~mod~} N$, and
\begin{equation*}
\left\{
\begin{array}{ll}
N + N_1 & \mbox{if~} \tau_0 = 0, \frac{N}{2} \\
2 N_2 & \mbox{otherwise}
\end{array}
\right.
\end{equation*}
if $e_2 - e_1 \equiv N / 2 \mathrm{~mod~} N$.
For $\tau = 2 \tau_0 + 1, 0 \leq \tau_0 \leq N - 1, \tau_0 \in \mathbb{Z}$,
the cross-correlation $R_{S^{e_1}, S^{e_2}} (\tau)$ of $S^{e_1}$ and $S^{e_2}$ is given as
\begin{equation*}
\left\{
\begin{array}{ll}
- N - N_2 & \mbox{if~} \tau_0 \equiv - e_2 - \frac{N}{2}, e_1 - 1 + \frac{N}{2} \mathrm{~mod~} N \\
- N_1 - N_2 & \mbox{if~} \tau_0 \equiv - e_2, e_1 - 1 \mathrm{~mod~} N \\
- 2 N_2 & \mbox{otherwise}
\end{array}
\right.
\end{equation*}
if $e_1 + e_2 \not\equiv 1 \mathrm{~mod~} N$ and $e_1 + e_2 \not\equiv 1 - N / 2 \mathrm{~mod~} N$,
\begin{equation*}
\left\{
\begin{array}{ll}
- 2 N & \mbox{if~} \tau_0 \equiv - e_2 - \frac{N}{2} \mathrm{~mod~} N \\
- 2 N_1 & \mbox{if~} \tau_0 \equiv - e_2 \mathrm{~mod~} N \\
- 2 N_2 & \mbox{otherwise}
\end{array}
\right.
\end{equation*}
if $e_1 + e_2 \equiv 1 \mathrm{~mod~} N$, and
\begin{equation*}
\left\{
\begin{array}{ll}
- N - N_1 & \mbox{if~} \tau_0 \equiv - e_2 - \frac{N}{2}, - e_2 \mathrm{~mod~} N \\
- 2 N_2 & \mbox{otherwise}
\end{array}
\right.
\end{equation*}
if $e_1 + e_2 \equiv 1 - N / 2 \mathrm{~mod~} N$.
Here $N_1$ and $N_2$ are defined as in (\ref{equation:AC of NTU}).
\end{theorem}

\begin{proof}
Exactly like the proof of Theorem \ref{theorem:periodic autocorrelation}.
\end{proof}

\begin{example}
Let $p = 11$ and $m = 2$, so that $N = 24$.
Let $\alpha$ be a root of the irreducible polynomial $x^2 + 7 x + 2$, and
$\omega = 9 + 2 \alpha \in \mathbb{F}_p ( \alpha )$.
Figure \ref{figure:CC of IS of NTU sequences in the case of p = 11, m = 2, e1 = 9, e2 = 11}
describes the graph of the periodic cross-correlation of $S^{9}$ and $S^{11}$.
Figure \ref{figure:CC of IS of NTU sequences in the case of p = 11, m = 2, e1 = 6, e2 = 18}
describes the graph of the periodic cross-correlation of $S^{6}$ and $S^{18}$.
Figure \ref{figure:CC of IS of NTU sequences in the case of p = 11, m = 2, e1 = 11, e2 = 14}
describes the graph of the periodic cross-correlation of $S^{11}$ and $S^{14}$.
Figure \ref{figure:CC of IS of NTU sequences in the case of p = 11, m = 2, e1 = 2, e2 = 11}
describes the graph of the periodic cross-correlation of $S^{2}$ and $S^{11}$.
\begin{figure}[htbp]
\begin{center}
\includegraphics[scale=0.85]{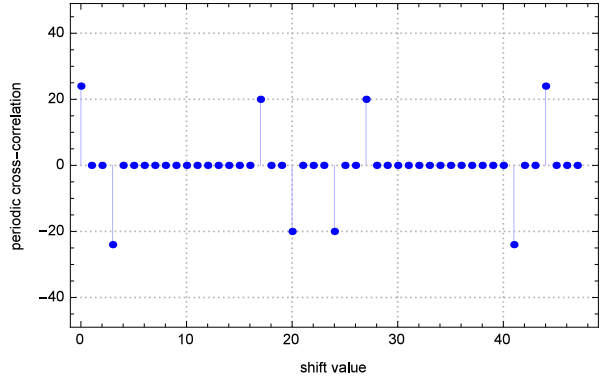}
\end{center}
\caption{The graph of the periodic cross-correlation of $S^{9}$ and $S^{11}$ in the case of $p = 11, m = 2$.}
\label{figure:CC of IS of NTU sequences in the case of p = 11, m = 2, e1 = 9, e2 = 11}
\end{figure}
\begin{figure}[htbp]
\begin{center}
\includegraphics[scale=0.85]{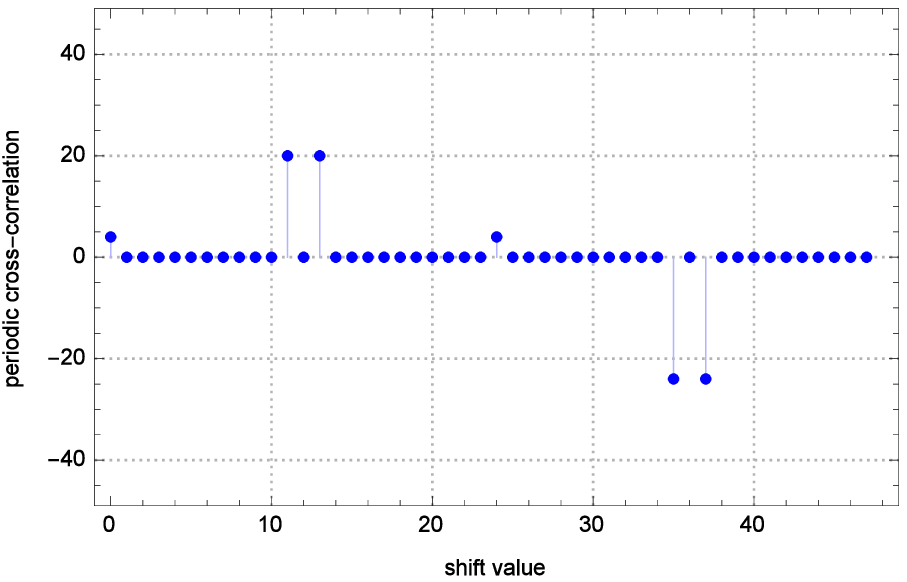}
\end{center}
\caption{The graph of the periodic cross-correlation of $S^{6}$ and $S^{18}$ in the case of $p = 11, m = 2$.}
\label{figure:CC of IS of NTU sequences in the case of p = 11, m = 2, e1 = 6, e2 = 18}
\end{figure}
\begin{figure}[htbp]
\begin{center}
\includegraphics[scale=0.85]{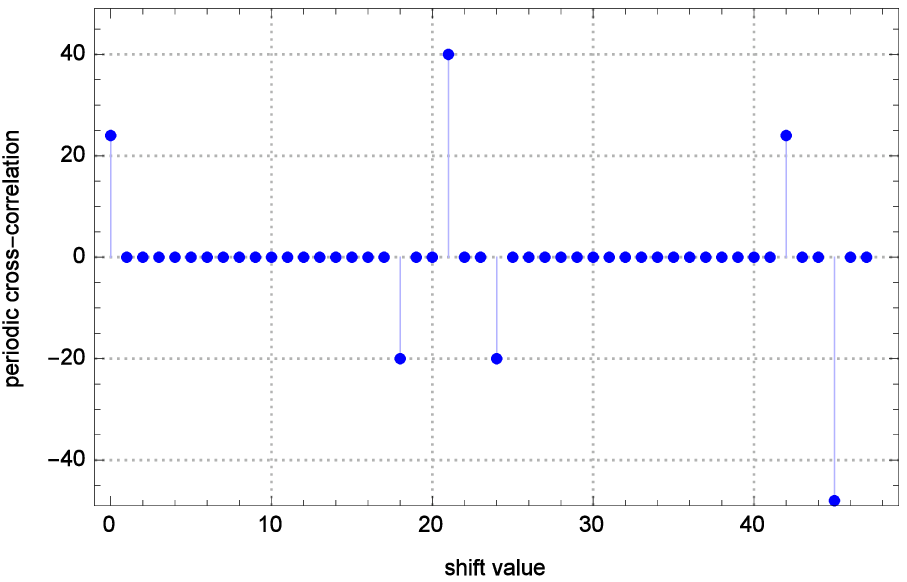}
\end{center}
\caption{The graph of the periodic cross-correlation of $S^{11}$ and $S^{14}$ in the case of $p = 11, m = 2$.}
\label{figure:CC of IS of NTU sequences in the case of p = 11, m = 2, e1 = 11, e2 = 14}
\end{figure}
\begin{figure}[htbp]
\begin{center}
\includegraphics[scale=0.85]{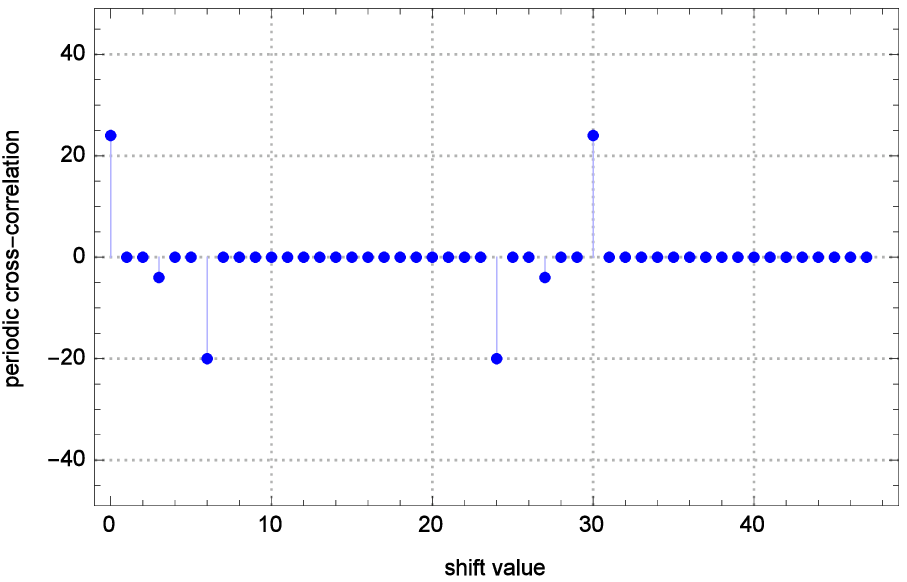}
\end{center}
\caption{The graph of the periodic cross-correlation of $S^{2}$ and $S^{11}$ in the case of $p = 11, m = 2$.}
\label{figure:CC of IS of NTU sequences in the case of p = 11, m = 2, e1 = 2, e2 = 11}
\end{figure}
\end{example}

\begin{example}
Let $p = 5$ and $m = 3$, so that $N = 62$.
Let $\alpha$ be a root of the irreducible polynomial $x^3 + 3 x^2 + 2 x + 3$, and
$\omega = 1 + \alpha + 2 \alpha^2 \in \mathbb{F}_p ( \alpha )$.
Figure \ref{figure:CC of IS of NTU sequences in the case of p = 5, m = 3, e1 = 25, e2 = 41}
describes the graph of the periodic cross-correlation of $S^{25}$ and $S^{41}$.
Figure \ref{figure:CC of IS of NTU sequences in the case of p = 5, m = 3, e1 = 4, e2 = 35}
describes the graph of the periodic cross-correlation of $S^{4}$ and $S^{35}$.
Figure \ref{figure:CC of IS of NTU sequences in the case of p = 5, m = 3, e1 = 26, e2 = 37}
describes the graph of the periodic cross-correlation of $S^{26}$ and $S^{37}$.
Figure \ref{figure:CC of IS of NTU sequences in the case of p = 5, m = 3, e1 = 45, e2 = 49}
describes the graph of the periodic cross-correlation of $S^{45}$ and $S^{49}$.
\begin{figure}[htbp]
\begin{center}
\includegraphics[scale=0.85]{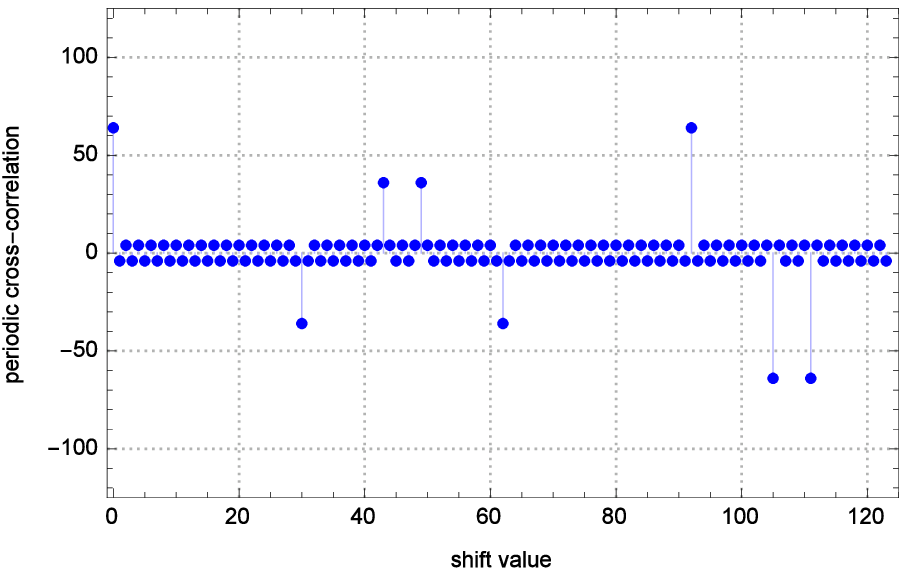}
\end{center}
\caption{The graph of the periodic cross-correlation of $S^{25}$ and $S^{41}$ in the case of $p = 5, m = 3$.}
\label{figure:CC of IS of NTU sequences in the case of p = 5, m = 3, e1 = 25, e2 = 41}
\end{figure}
\begin{figure}[htbp]
\begin{center}
\includegraphics[scale=0.85]{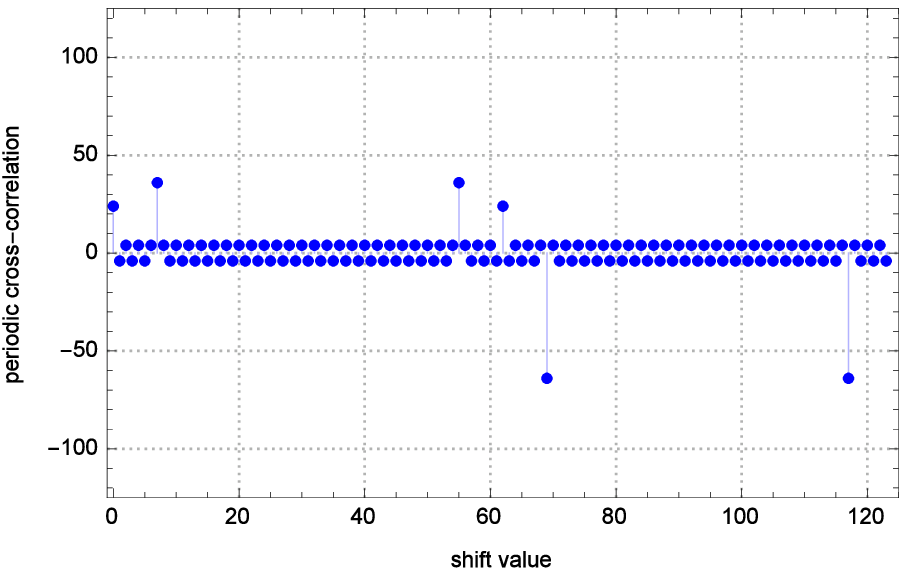}
\end{center}
\caption{The graph of the periodic cross-correlation of $S^{4}$ and $S^{35}$ in the case of $p = 5, m = 3$.}
\label{figure:CC of IS of NTU sequences in the case of p = 5, m = 3, e1 = 4, e2 = 35}
\end{figure}
\begin{figure}[htbp]
\begin{center}
\includegraphics[scale=0.85]{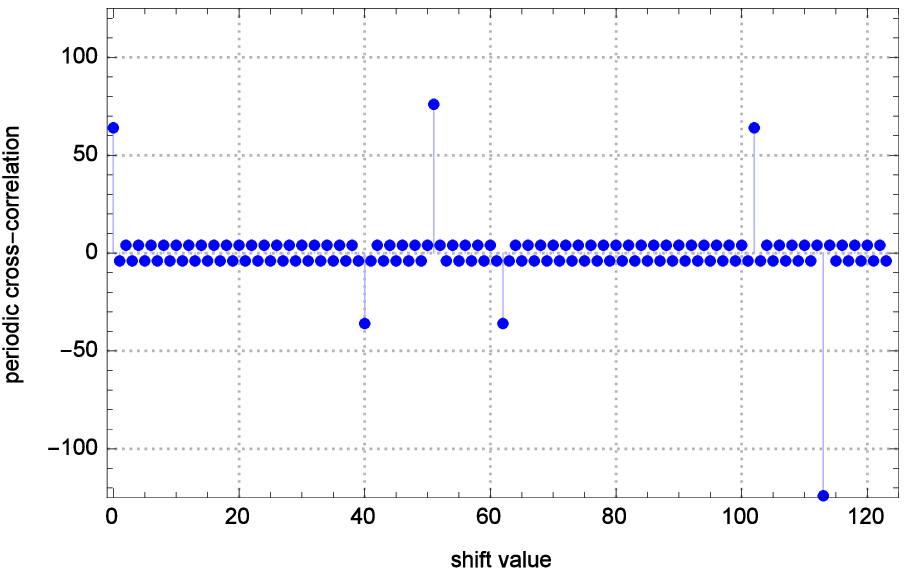}
\end{center}
\caption{The graph of the periodic cross-correlation of $S^{26}$ and $S^{37}$ in the case of $p = 5, m = 3$.}
\label{figure:CC of IS of NTU sequences in the case of p = 5, m = 3, e1 = 26, e2 = 37}
\end{figure}
\begin{figure}[htbp]
\begin{center}
\includegraphics[scale=0.85]{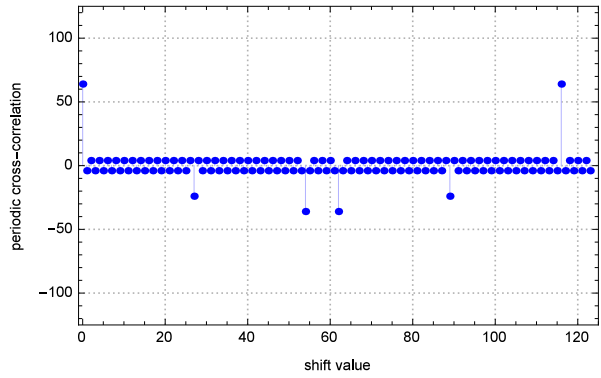}
\end{center}
\caption{The graph of the periodic cross-correlation of $S^{45}$ and $S^{49}$ in the case of $p = 5, m = 3$.}
\label{figure:CC of IS of NTU sequences in the case of p = 5, m = 3, e1 = 45, e2 = 49}
\end{figure}
\end{example}

\subsection{Linear complexity of the proposed sequences}

The minimal polynomial and linear complexity of the proposed sequence are given as follows:

\begin{theorem}
Let $p > 2$ be a prime number and $m > 1$ an integer. 
Let $\omega \in \mathbb{F}_{p^m}$ be a primitive element in $\mathbb{F}_{p^m}$.
Let $T^{(1)}$ and $T^{(2)}$ be the sequences defined as (\ref{equation:1st type NTU}) and (\ref{equation:2nd type NTU}) with respect to $p, m$ and $\omega$, respectively.
For $e \in \{ 0, 1, \dots , N - 1 \}$,
let $S^{e}$ be the interleaved sequence of $\{ T^{(1)}, L^{e} (T^{(2)}) \}$.
Put $N = 2 (p^m - 1) / (p - 1)$.
Then the minimal polynomial $m(x) \in \mathbb{F}_2 [x]$ of $S^{e}$ is given as
\begin{equation*}
m(x) = \frac{x^{2N} - 1}{x^{G ( N, e )} - 1},
\end{equation*}
where $G ( N, e ) = \mathrm{GCD} (N / 2^{\nu_2 (N)}, - 2 e + 1 ~\underline{\mathrm{mod}}~ N / 2^{\nu_2 (N)})$ and
$\nu_2 (N)$ is the exponent of the largest power of $2$ which divides $N$.
Therefore the linear complexity $L( S^{e} )$ of $S^{e}$ is given as
\begin{equation*}
L( S^{e} ) = 2 N - G ( N, e ).
\end{equation*}
\end{theorem}

\begin{proof}
Since
\begin{eqnarray*}
L^{e} ( T^{(2)} ) (x)
& \equiv & x^{N - e} T^{(2)} (x) \mathrm{~mod~} x^N - 1 \\
& \equiv & x^{N - e} \left\{ x^{\frac{N}{2}} T^{(1)} (x) + \frac{x^N - 1}{x - 1} \right\} \mathrm{~mod~} x^N - 1 \\
& \equiv & x^{\frac{3 N}{2} - e} T^{(1)} (x) + x^{N - e} \cdot \frac{x^N - 1}{x - 1} \mathrm{~mod~} x^N - 1,
\end{eqnarray*}
we have
\begin{eqnarray*}
S^{e} (x) & = & T^{(1)} (x^2) + x L^{e} ( T^{(2)} ) (x^2) \\
& \equiv & T^{(1)} (x^2) + x \left\{ x^{3 N - 2 e} T^{(1)} (x^2)
+ x^{2 N - 2 e} \cdot \frac{x^{2 N} - 1}{x^2 - 1} \right\} \mathrm{~mod~} x^{2 N} - 1 \\
& \equiv & ( x^{3 N - 2 e + 1} + 1 ) T^{(1)} (x^2) + x^{2 N - 2 e + 1} \cdot \frac{x^{2 N} - 1}{x^2 - 1}
\mathrm{~mod~} x^{2 N} - 1.
\end{eqnarray*}
Note that
\begin{enumerate}
\item
Since $L( T^{(1)} ) = N$, $\mathrm{GCD} (x^N - 1, T^{(1)} (x)) = 1$.
\item
Since $3 N - 2 e + 1 \equiv 1 \mathrm{~mod~} 2$, $x^2 - 1$ does not divide $x^{3 N - 2 e + 1} - 1$.
\end{enumerate}
Hence it follows that
\begin{eqnarray*}
\mathrm{GCD} (x^{2 N} - 1, S^{e} (x))
& = & \mathrm{GCD} (x^{2 N} - 1, x^{3 N - 2 e + 1} + 1) \\
& = & x^{\mathrm{GCD} (2 N , 3 N - 2 e + 1)} - 1 \\
& = & x^{G ( N, e )} - 1.
\end{eqnarray*}
Therefore we have
\begin{equation*}
m ( x ) = \frac{x^{2N} - 1}{\mathrm{GCD} (x^{2 N} - 1, S^{e} (x))} = \frac{x^{2N} - 1}{x^{G ( N, e )} - 1},
\end{equation*}
and
\begin{eqnarray*}
L( S^{e} ) & = & 2 N - \mathrm{deg} \, \mathrm{GCD} (x^{2 N} - 1, S^{e} (x)) \\
& = & 2 N - \mathrm{deg} \, \left( x^{G ( N, e )} - 1 \right) \\
& = & 2 N - G ( N, e ).
\end{eqnarray*}
\end{proof}

\begin{corollary}
The upper and lower bounds on the linear complexity are obtained as follows:
\begin{enumerate}
\item $L( S^{e} ) \leq 2 N - 1$.
The equality holds if and only if $G ( N, e ) = 1$.
\item $L( S^{e} ) \geq 2 N - N / 2^{\nu_2 (N)}$.
The equality holds if and only if $- 2 e + 1 \equiv 0 \mathrm{~mod~} N / 2^{\nu_2 (N)}$.
\end{enumerate}
\end{corollary}

\begin{example}
Let $p = 11$ and $m = 2$, so that $N = 24$.
Let $\alpha$ be a root of the irreducible polynomial $x^2 + 7 x + 2$, and
$\omega = 9 + 2 \alpha \in \mathbb{F}_p ( \alpha )$.
Then
the linear complexity $L( S^{e} )$ of $S^{e}$ is given as
\begin{equation*}
L( S^{e} ) =
\left\{
\begin{array}{ll}
47 & \mbox{if~} e \equiv 0, 1 \mathrm{~mod~} 3 \\
45 & \mbox{if~} e \equiv 2 \mathrm{~mod~} 3.
\end{array}
\right.
\end{equation*}
Note that $N = 24 = 2^3 \cdot 3$.
\end{example}

\begin{example}
Let $p = 5$ and $m = 3$, so that $N = 62$.
Let $\alpha$ be a root of the irreducible polynomial $x^3 + 3 x^2 + 2 x + 3$, and
$\omega = 1 + \alpha + 2 \alpha^2 \in \mathbb{F}_p ( \alpha )$.
Then
the linear complexity $L( S^{e} )$ of $S^{e}$ is given as
\begin{equation*}
L( S^{e} ) =
\left\{
\begin{array}{ll}
93 & \mbox{if~} e = 16, 47 \\
123 & \mbox{otherwise}.
\end{array}
\right.
\end{equation*}
Note that $N = 62 = 2 \cdot 31$ and $- 2 \cdot 16 + 1 \equiv -2 \cdot 47 + 1 \equiv 0 \mathrm{~mod~} 31$.
\end{example}

\section{Conclusion}
\label{section:Conclusion}

In this paper, we propose interleaved sequences of geometric sequences of two types,
which were introduced by Nogami, Tada and Uehara.
Furthermore, we show correlation properties and linear complexity of the interleaved sequences.
The proposed sequences have the balance property and
double the period of the geometric sequences by the interleaved structure.
The autocorrelation distributions of the proposed sequences have peaks and troughs
at six shift values (four shift values for special cases).
The cross-correlation distributions of the two proposed sequences that have different shift widths
also have a few peaks and troughs.
The question is whether a security is affected by these distributions.
We obtain the linear complexity of the proposed sequences.
The formula induces the upper and lower bounds on the linear complexity,
and the conditions to attain the upper and lower bounds, respectively.
Therefore, although the linear complexity do not attain the maximal value,
one can choose a shift width with which a sequence has a large linear complexity,
especially such as the period minus one.
Thus the proposed sequences have good cryptographic properties for linear complexity.
We only consider the frequency for a single symbol,
however the proposed sequences do not have the balance property for block of length larger than one.
As a future work, we should give a transformation,
by which transformed sequences have the balance property for block of a certain length
and keep some good properties of the propesed sequences.

\section*{Acknowledgments}

%The authors would like to thank the anonymous reviewers for helpful comments and suggestions.
This research was supported by JSPS KAKENHI Grant-in-Aid for Scientific Research (A) Number 16H01723.

\bibliographystyle{abbrv}% bib style
\bibliography{InterleavedNTU}% your bib database
%\begin{thebibliography}{99}% more than 9 --> 99 / less than 10 --> 9
%\bibitem{}
%\end{thebibliography}

\end{document}